\documentclass{svmult}
  
\usepackage{mathptmx,hyperref}   
\usepackage{helvet} 
\usepackage{courier}
\usepackage{amsmath,amsfonts}
\usepackage{graphicx}
\usepackage[bottom]{footmisc}          
    
\usepackage{mathptmx}         
\usepackage{helvet}            
\usepackage{courier}            
\usepackage{type1cm}             
\usepackage{makeidx}             
\usepackage{graphicx}          
\usepackage{multicol}          
\usepackage[bottom]{footmisc}  
     
\usepackage{epsfig}     
\usepackage{psfrag,rotating}     
\usepackage{amssymb,floatflt,enumerate} 
\usepackage{amsmath,amscd,psfrag,leqno} 
\usepackage[mathscr]{eucal}  
\usepackage{shadethm}           
\usepackage{graphicx}   
\usepackage{overpic,contour}       
\contourlength{0.3mm}  
 
\def\Beweisende{\square}            
\def\BewEnde{\hfill{\Beweisende}}

\def\phm{{\hphantom{-}}}

\def\NN{{\mathbb N}}

\def\Vkt#1{{\mathbf #1}}

\usepackage[dvipsnames]{xcolor}

\makeindex
\begin{document}

\title*{Flexible discrete translational surfaces}

% Use \titlerunning{Short Title} for an abbreviated version of
% your contribution title if the original one is too long
\author{Georg Nawratil}
\authorrunning{G. Nawratil}
% Use \authorrunning{Short Title} for an abbreviated version of
% your contribution title if the original one is too long
\institute{
  Institute of Discrete Mathematics and Geometry \&  
	Center for Geometry and Computational Design, TU Wien, Austria \\
    \email{nawratil@geometrie.tuwien.ac.at}
   }

%
% Use the package "url.sty" to avoid
% problems with special characters
% used in your e-mail or web address
%
\maketitle

\abstract{We give a full list of translational nets 
which flex within their class of discrete surfaces of translation, by reducing the classification problem to the one of flexible complete bipartite frameworks on the sphere, for which the solution is known. We also obtained two novel classes which correspond to Bottema's spherical 16-bar mechanisms and the constant diagonal angle frameworks. Based on an algorithm for the construction of all flexible translational nets, we also discuss flexible translational tubes and toroids. Furthermore, we present novel results for both topologies which are implied by  Bottema's spherical 16-bar mechanisms. 
}

\section{Introduction and review}\label{sec:intro}

In the last decades, interest in flexible quad surfaces has increased due to their relevance for the transformable design in the context of rigid origami. More precisely, we consider nets of planar quadrilateral (PQ) faces, that are hinged by rotational joints in the combinatorics of a square grid. The faces have to be strictly convex or strictly non-convex in order to guarantee their quadrilateral shape; i.e.\ no degenerations into triangles, line-segments or points are allowed.

These PQ-nets are in general rigid, but certain geometries allow for a finite flex of the structure, which is also known as rigid-foldability, isometric deformation or continuous flexibility. In the remainder of this paper, we limit ourselves to the terms {\it flexible, flexibility} and  {\it flex}, which are always meant in context of finite mobility.

In the generic case the flexibility of a PQ-net results form the flexibility of each $(3\times 3)$ subnet as originally stated in \cite{BHS08}, but with an erroneous formulation of the degenerate cases. This was corrected in the recent paper \cite{sobhan}, where it was pointed out that genericity is guaranteed if the PQ-net has no pair of coplanar adjacent faces. In more detail this theorem reads as follows:

\begin{theorem} \label{thm:sobhan} \cite[Theorem 5]{sobhan}
If a $(m\times n)$ PQ-net has no pair of coplanar adjacent faces and all of its
$(3\times 3)$ subnets are flexible, then the net is flexible as a whole. 
\end{theorem}

Therefore the classification of these flexible $(3\times 3)$ PQ-nets and their geometric understanding plays a central role in the theory of flexible discrete surfaces. Their systematic study dates back to Kokotsakis, who listed in 
\cite[§16--§18]{kokotsakis} the following cases: 
\begin{enumerate}[(a)]
    \item 
    The V-hedral\footnote{All inner vertices of the $(3\times 3)$ PQ-net are V-hedral, which means that opposite angles at the vertex are equal.} case already known to Sauer and Graf \cite{sauer_graf}, but he also mentioned for the first time the anti-V-hedral\footnote{All inner vertices of the $(3\times 3)$ PQ-net are anti-V-hedral, which means that opposite angles at the vertex are supplementary.} case and the two possible hybrid\footnote{Consisting of V-hedral and anti-V-hedral vertices.} cases. 
    \item 
    The T-hedral case already known to Sauer and Graf \cite{sauer_graf}.
    \item 
    The developable case of the line-symmetric type of Stachel \cite{stachel2}  also known as Kokotsakis tilling.
    \item 
    The planar symmetric case, which is a trivial one like the translational case mentioned by Stachel \cite[Sec.\ 1.2]{stachel2}. 
\end{enumerate}
Based on spherical kinematic geometry \cite{stachel2} a partial classification of flexible $(3\times 3)$ PQ-nets was obtained by Stachel and the author \cite{Naw11,Naw12,NS10}. 
Inspired by this approach, Izmestiev \cite{Izm17} obtained a full classification in 2017. Nevertheless the complete list of cases is known, there are only a few results on flexible PQ-nets larger than $(3\times 3)$-patches  available in the literature until now: 
\begin{enumerate}[(A)]
    \item 
    It is known that V-hedral, anti-V-hedral and hybrid nets of infinite size can be constructed. For an overview on these well studied PQ-nets we refer to  \cite{kilian}.
    \item 
    PQ-nets of arbitrary size can also be constructed within the classes of T-hedra \cite{kiumars3} and P-hedra \cite{phedra2}, which can both be seen as generalizations of flexible axial cone-nets studied in \cite{phedra1}.  
    \item 
    Some initial stitching solutions of different types of flexible $(3\times 3)$ PQ-nets were given by He and Guest \cite{HG18}. Recently, a 
    repetitive stitching construction method has been presented in \cite{he26} that enables an infinite mesh refinement for two families of PQ-nets generated from linear and equimodular couplings of Izmestiev's classification.  The pattern constructed by equimodular couplings generalizes the approach presented in \cite{Dieleman}, which already includes the Kokotsakis tiling discussed in detail by Huffman \cite{huffman} and Stachel \cite{tile}.
\end{enumerate}
Note that the flexibility of PQ-nets is not a property of the extrinsic
geometry but of the intrinsic one, which is determined by the corner angles of the
planar quads only. Nonetheless, the discrete quad surfaces listed in (A) and (B) allow direct access to their spatial shape through the use of control polylines (see e.g.\ \cite{kilian,phedra2,kiumars3}). Due to this intuitive design methods these classes are well suited for interactive transformable design.

The classes of V-hedra and T-hedra have also flexible smooth counterparts, which are the Voss surfaces \cite{voss_orig} (smooth V-surfaces) and profile affine surfaces \cite{sauer,sauer_graf}, respectively, which we call smooth T-surfaces in the remainder of this paper.  For a contemporary study of the relation between the smooth and the discrete case of these flexible surfaces we refer to the works \cite{ivan2,ivan1,tubes}, which also discuss the intermediate case of semi-discretization. 
In contrast, P-hedra only have semi-discrete analogues but no smooth ones due to their construction outlined in \cite{phedra2}, which is based on the consecutive application of projective generalizations of the two trivial cases listed in (d). As a consequence, one of the input polylines is composed of line-segments parallel to only four directions, while the other one can be smoothed.

\subsection{Motivation and outline}

The class of T-surfaces can be divided into several subclasses \cite{ivan1,sauer,sauer_graf}, which include, among others, rotational surfaces, moulding surfaces and translational surfaces. These are kinematical surfaces as they result from the sweep of a rigid curve under a continuous Euclidean motion.  It is remarkable that a T-surface remains within its subclass during the associated flex, which holds true in the smooth, discrete and semi-discrete setting. Note that this property is not self-evident, which can be seen by rotational V-surfaces \cite{ivan2,tachauer}, e.g.\ the catenoid flexes into the right helicoid.  

From now on, we want to focus on translational surfaces, which can be generated by translating one space curve ({\it generatrix}, which is called {\it profile}  in case of a planar curve) along another space curve ({\it trajectory}). It is well-known that the same surface is swept if the two curves change their roles. 
In the discrete setting the trajectory and generatrix are two polylines $\Vkt t=(\Vkt t_0,\ldots ,\Vkt t_m)$ and $\Vkt g=(\Vkt g_0,\ldots ,\Vkt g_n)$ with $m\in\NN$ and $n\in\NN$, respectively, vertices. By applying the discrete translations implied by one of the polylines to the other one, we obtain a discrete surface with parallelogram  faces. Therefore a discrete translational surface is automatically a PQ-net. Therefore, we can restrict to the term {\it translational net} in the remainder of the paper.
 
Note that a translational $(m\times n)$ net keeps the intrinsic property of being flexible under the so-called {\it parallelism operation}, which was already used in \cite{sauer,sauer_graf}. In general it can be applied to any flexible PQ-net, but in case of translational surfaces it boils down to the change of lengths of the line-segments of the two input polylines $\Vkt t$ and $\Vkt g$. This procedure even includes negative lengths, which correspond  to the change of the orientation of the line-segment into the opposite direction. Moreover, we can assume $m,n>0$ and $(m,n)\neq (1,1)$, because for $m=n=1$ we end up with a rigid parallelogram.

To the best of the author's  knowledge, the following two classes of flexible translational $(m\times n)$ nets are reported in the literature so far (cf.\ Nassar \cite{nassar}):
\begin{enumerate}
    \item [$N_0$:]
    {\bf Trivially flexible translational nets:} These translational nets are generated by an arbitrary trajectory and a zigzag profile with only two slopes (i.e.\ every second line-segment is parallel). Its flexibility results from the fact that each  $(m\times 3)$-strip belongs to the trivial translational case mentioned in (d) taking the parallelism operation into account. As this property is preserved during the flex, this PQ-net deforms within the class of discrete translational surfaces. As already noted above this also holds true for the next class.
     \item [$N_1$:]
    {\bf Translational T-hedra:} These flexible nets are obtained if the generatrix and the trajectory are both planar and are located in orthogonal planes \cite{ivan1,sauer,sauer_graf}. 
\end{enumerate} 
No further contemporary works on this topic are known to the author, which dates back to \cite[§33]{peterson} according to Voss's encyclopedic article \cite[Nr.\ 26]{voss_ency}, where further historical references regarding the smooth case can be found. 
This raises the question if further flexible translational nets exist. In the paper at hand, which is structured as follows, we give a full classification of these PQ-nets which flex within the class of discrete surfaces of translation: 

In Section \ref{sec:equi} we show that this classification problem is equivalent to the one of flexible complete bipartite framework on the sphere, for which the solution is known. Based on this we present in Section \ref{sec:construct} the construction of all flexible translational nets using two input polylines, which also extends the list of known flexible PQ-nets of infinite size given in items (A,B,C) of Section \ref{sec:intro}. 
Especially, we discuss flexible translational tubes and toroids in the Sections \ref{sec:tube} and \ref{sec:toroid}, respectively. Finally, we conclude the paper in Section \ref{sec:conclusion}.

\section{Equivalent spherical problem}\label{sec:equi}

In this section we transfer the problem of determining all flexible translational quad-nets which flex within the class of discrete surfaces of translation to the sphere. This is motivated by the following theorem of Stachel \cite{stachel2}:

\begin{theorem}\label{thm:stachel}  \cite[Theorem 1]{stachel2}
A $(3\times 3)$ PQ-net is flexible if and only if its spherical image is flexible.
\end{theorem}

At this point we have to explain how the mentioned spherical image is obtained:  One shifts all edges of the PQ-net through the origin and intersects the resulting set of lines with the unit sphere centered in the origin. Each line gives rise to a pair of antipodal points on this sphere. According to \cite{stachel2} it does not matter which of the two points we consider as its spherical image. The face spanned by two edges sharing the same vertex of the PQ-net is mapped to the arc of a great circle between the corresponding  spherical image points of the two edges. Moreover, we can assume without loss of generality that its arc length is within the interval $(0;\pi)$, which is open 
due to the basic assumption on PQ-nets that the faces have to be strictly convex or strictly non-convex (cf.\ first paragraph of Section \ref{sec:intro}).

Theorem \ref{thm:sobhan} and Theorem \ref{thm:stachel} imply together the following statement: 

\begin{corollary}
A $(m\times n)$ PQ-net with no pair of coplanar adjacent faces is flexible if and only if its spherical image is flexible.
\end{corollary}

Clearly, this statement also holds for the subclass of flexible translational nets, but they have a very special spherical image as pointed out in the next theorem:

\begin{theorem}\label{thm:equi}
A translational net with no pair of coplanar adjacent faces can flex within the class of discrete surfaces of translation if and only if its spherical image is flexible. Without loss of generality we can assume that this image is a nondegenerate  nonoverlapping complete bipartite framework on the sphere. Moreover, the flex of the translational surface has to be 1-dimensional.
\end{theorem}

Before proceeding to the proof, we define the new terms used in this theorem:

\begin{definition}
A complete bipartite $(m,n)$ framework on the sphere consists of two sets of points
$\mathcal{P}=(\Vkt p_1,\ldots,\Vkt p_m)$ and  
$\mathcal{Q}=(\Vkt q_1,\ldots,\Vkt q_n)$ on the sphere. Moreover, every point $\mathcal{P}$ is connected with every point of $\mathcal{Q}$ by a spherical bar. This spherical framework is called: 
\begin{enumerate}[$\phm\bullet$]
    \item 
    {\it nondegenerate} if no bar has zero length or length $\pi$; i.e.\ 
    \begin{equation}
    \Vkt p_i\neq \pm\Vkt q_j \quad \text{for}\quad \forall i,j
    \end{equation}
    where $\pm$ indicate the antipodal point pair;
    \item
    {\it nonoverlapping} if not two points of a set coincide or are antipodal; i.e.\
    \begin{equation}
    \begin{split}
    &\Vkt p_i\neq \pm\Vkt p_j \quad \text{for}\quad i\neq j \\
    &\Vkt q_i\neq \pm\Vkt q_j \quad \text{for}\quad i\neq j 
    \end{split}
    \end{equation}
\end{enumerate}
\end{definition}

For proving Theorem \ref{thm:equi} we have to do some preparatory work, which is based on the following known result:

\begin{theorem} \label{thm:spezi} \cite[Theorem 3.1]{BHS08}
    A PQ-net admits at most a 1-parameter flex if not every vertex contains at least a pair of collinear opposite edges. 
\end{theorem}

\begin{lemma}\label{lem:spezi}
    A translational net, where every vertex contains at least a pair of collinear opposite edges, has to have either the generatrix or the trajectory or both as straight lines. 
\end{lemma}

\begin{proof}
    The proof is done in an indirect way. Assume both are not straight lines; 
    i.e.\  there exist vertices $\Vkt t_i$ of the trajectory and $\Vkt g_j$ of the generatrix where the two adjacent line-segments are not collinear. By construction the vertex $X_{i,j}$ of the translational net does not have at least one pair of collinear opposite edges. \hfill $\BewEnde$
\end{proof}

\noindent
{\bf Proof of Theorem \ref{thm:equi}:} 
Due to Lemma \ref{lem:spezi} and the assumption that the translational net has no pair of coplanar adjacent faces, we can conclude from Theorem \ref{thm:spezi}, that the flex is 1-dimensional. Therefore this also holds for the flex of the spherical image, which is constructed next.

To do so, we consider the translational $(m,n)$ net with 
trajectory polyline $\Vkt t=(\Vkt t_0,\ldots ,\Vkt t_m)$ and 
generatrix polyline  $\Vkt g=(\Vkt g_0,\ldots ,\Vkt g_n)$.  
By applying the construction of the spherical image each line-segment bounded by the $(i-1)$-th and $i$-th vertex is mapped to one of the points:
\begin{equation}
\begin{split}
    &\pm \Vkt p_i:=\pm \tfrac{\Vkt t_i-\Vkt t_{i-1}}{\| \Vkt t_i-\Vkt t_{i-1} \|}\quad 
    \text{for}\quad i=1,\ldots, m \\
    &\pm \Vkt q_i:=\pm \tfrac{\Vkt g_i-\Vkt g_{i-1}}{\| \Vkt g_i-\Vkt g_{i-1} \|}\quad 
    \text{for}\quad i=1,\ldots, n
\end{split}
\end{equation}
In this way we get the two sets of spherical points $\mathcal{P}=(\Vkt p_1,\ldots,\Vkt p_m)$ and  $\mathcal{Q}=(\Vkt q_1,\ldots,\Vkt q_n)$. The complete bipartite combinatorics follows from the construction of the translational surfaces; i.e.\ every line-segment of $\Vkt t$ generates with every line-segment of $\Vkt g$ one face of the $(m,n)$ net, which is mapped to the corresponding great circular arc on the sphere. 
Moreover, as each face has to be a proper parallelogram (cf.\ first paragraph of Section \ref{sec:intro}), the resulting complete bipartite spherical $(m,n)$ framework is nondegenerate. 

We remain to show the existence of a nonoverlapping configuration. 
In general if a nondegenerate complete bipartite $(m,n)$ framework on the sphere is flexible then there are only finitely many configurations  where the $i$-th point and the $j$-th one (with $i\neq j$) of the same set coincide or are antipodal.  
As a consequence the resulting nondegenerate complete bipartite $(m,n)$ spherical framework has to have infinitely many nonoverlapping configurations and we are done. 
Therefore we remain with the discussion of nondegenerate complete bipartite $(m,n)$ frameworks which have infinitely many configurations where the $i$-th point and the $j$-th one (with $i\neq j$) of the same set are identical or antipodal. We have to distinguish two cases:

\begin{enumerate}[(i)]
    \item  The $i$-th point and the $j$-th one (with $i\neq j$)  of the same set are identical or antipodal during the whole flex. In this case the two points can be identified an we reduce the spherical framework to be a complete bipartite one of type $(m-1,n)$ or $(m,n-1)$, respectively. We can repeat this reduction operation until we remain with a  spherical complete bipartite $(k,l)$ framework with $0<k\leq m$, $0<l\leq n$ and $(k,l)\neq (1,1)$, which has no two points of the same set identical or antipodal during the whole flex. If the reduced spherical framework has a nonoverlapping configuration, we are done. If such a configuration does not exist we proceed with item (ii).
    \item
    Due to item (i) we can assume that the spherical complete bipartite $(k,l)$ framework with $0<k\leq m$, $0<l\leq n$ and $(k,l)\neq (1,1)$, does not have two points of the same set, which are identical or antipodal during the whole flex,  but it still does not possess a nonoverlapping  configuration.
    Theoretically this is possible as a 1-dimensional flex can have different branches. For example let us consider a translational T-hedron with a zigzag profile polyline, where every second line-segment is parallel to the horizontal base plane containing an arbitrary trajectory polyline.  Each of the respective horizontal strips is in a bifurcation configuration and can flex up or down, which corresponds to different branches of the 1-dimensional flex, also known as motion modes. 
    
    It is well known, that a spherical 4-bar mechanism can be assembled in two ways, which correspond to two motion modes. The change from one mode into the other can only be done if two neighbouring bars are on a great circle. But this would correspond to the case that two adjacent faces of the underlying translational net are coplanar, which is not possible due to our assumption. 
    Therefore the spherical complete bipartite $(k,l)$ framework is not in a bifurcation configuration. Therefore  we can apply the reduction operation of item (i) to  two points of the same set, which are identical or antipodal during the motion mode the configuration belongs to. By an iterative reduction we end up with a spherical complete bipartite $(e,f)$ framework with $0<e\leq k$, $0<f\leq l$ and $(e,f)\neq (1,1)$, which can flex along the branch into a nonoverlapping configuration, as there are only finitely many configurations left where two points of the same set are identical or antipodal. \hfill $\BewEnde$
\end{enumerate}

Due to Theorem \ref{thm:equi} the classification of flexible translational nets can be reduced to the one of flexible complete bipartite frameworks on the sphere. This problem was solved by Kovalev and Orevkov in \cite{kovalev} in 2023. Based on their result the classification can be formulated as follows:

\begin{theorem}\label{them:list}
A nondegenerate nonoverlapping bipartite $(e,f)$ framework on the sphere with $0<e,f\in\NN$ and $(e,f)\neq (1,1)$ is flexible if and only if it is one of the following cases: 
\begin{enumerate}
\item [$F_0$:]
{\bf Trivial flexible framework:} $e<3$ and/or $f<3$.  
\item [$F_1$:]
{\bf 1st spherical Dixon framework:} $\Vkt p_1,\ldots,\Vkt p_e$ are coplanar and
$\Vkt q_1,\ldots,\Vkt q_f$ are coplanar and the two carrier planes are orthogonal. 
\item [$F_2$:]
{\bf 2nd spherical Dixon framework:} There are two orthogonal planes $\alpha$ and $\beta$ passing through the center of the sphere and there are two rectangles $P$ and $Q$, which are symmetric with respect to $\alpha$ and $\beta$ and have all their vertices on the sphere but none of them lies in $\alpha$ or $\beta$. 
We get a spherical Dixon framework of the second type if $\mathcal{P}\in P$ and $\mathcal{Q}\in Q$ holds after applying a possibly necessary 
exchange of points $\Vkt p_i$ and $\Vkt q_j$ by their antipodes. This holds for $2<e,f<5$.
\item [$F_3$:]
{\bf Constant diagonal angle framework:} After a possibly necessary renumbering of vertices, the following algebraic relations  have to hold for $e=f=3$, which are formulated using the standard scalar product $\langle \cdot, \cdot\rangle$:
\begin{equation}
\begin{split}
    &\langle \Vkt p_2,\Vkt q_2\rangle =  \langle \Vkt p_3,\Vkt q_2\rangle =  \langle \Vkt p_3,\Vkt q_3\rangle =  -\langle \Vkt p_2,\Vkt q_3\rangle  \\
    &\langle \Vkt p_1,\Vkt q_1\rangle\neq 0\quad\text{and}\quad
    \langle \Vkt p_1,\Vkt q_k\rangle = \langle \Vkt p_k,\Vkt q_1\rangle =0 \quad \text{for}\quad k=2,3.
\end{split}
\end{equation}
\end{enumerate}
\end{theorem}

Note that the existence of the spherical analogues of both Dixon 9-bar linkages ($F_1$ and $F_2$), where the second one can be extended to Bottema's 16-bar mechanism, was already proven by Wunderlich \cite{wunderlich}. The entry $F_3$ was first listed in \cite{gallet}, where a full classification of flexible complete $(3,3)$ bipartite spherical frameworks was given. Therein the flex of $F_3$ was denoted by {\it constant diagonal angle motion}, as the angle enclosed by the diagonals of the spherical quadrilateral $\Vkt p_2,\Vkt q_2,\Vkt p_3,\Vkt q_3$ remains constant. 
We added to the characterisation of $F_3$ the condition $\langle \Vkt p_1,\Vkt q_1\rangle\neq 0$ because otherwise we end up with a special case of $F_1$.

\section{Construction of all flexible translational nets}\label{sec:construct}

In the following we give a 3-step algorithm for the construction of all translational $(m\times n)$ nets, 
which can flex within the class of discrete surfaces of translation: 

\begin{enumerate}[(1)]
    \item 
    We start with a nondegenerate nonoverlapping bipartite $(e,f)$ framework $F_i$ on the sphere listed in Theorem \ref{them:list}. Moreover, according to the proof of Theorem \ref{thm:equi} we can assume without loss of generality that it is not in a bifurcation configuration. 
    \item 
    Then we compose the trajectory polyline $\Vkt t=(\Vkt t_0,\ldots ,\Vkt t_m)$ in a way that every line-segment is parallel to a vector $\pm \Vkt p_i$. 
    The same holds for the generatrix polyline  $\Vkt g=(\Vkt g_0,\ldots ,\Vkt g_n)$ with respect to $\pm \Vkt q_j$. 
    \item 
    By translating one polyline along the other, we obtain the flexible  translational net $N_i$.
\end{enumerate}

Starting in step (1) with $F_0$ we obtain the trivially flexible translational nets $N_0$; for $F_1$ we end up with translational T-hedra $N_1$. The inputs  $F_2$ and $F_3$ result in novel flexible translational nets, which we call {\it  Bottema's spherical 16-bar based flexible translational nets} $N_2$ and {\it constant diagonal angle flexible translational nets} $N_3$, respectively.

If we add to step (2) the condition that no three consecutive points of the polylines $\Vkt t$ or $\Vkt g$ are allowed to be collinear, then the 
configuration constructed in step (3) has no coplanar adjacent faces. 
But can all flexible translational nets be constructed with the above given 3-step algorithm or have we missed some due to the assumption of Theorem \ref{thm:equi} that no adjacent faces are coplanar? The following theorem will show that the given construction is indeed exhaustive.

\begin{theorem}\label{thm:every}
    Every  translational $(m\times n)$ net, which can flex within the class of discrete surfaces of translation, can be constructed by the given 3-step algorithm.
\end{theorem}

\noindent
{\bf Proof of Theorem \ref{thm:every}:} 
Clearly, every translational net generated by the above 3-step algorithm is flexible within the class of discrete surfaces of translation due to its construction. 

We only have to show that also all translational $(m\times n)$ nets, which does not possess a configuration without two adjacent faces coplanar during the flex within the class of discrete surfaces of translation, can be produced in this way. For translational nets, where every vertex contains at least a pair of collinear opposite edges, we know due to Lemma \ref{lem:spezi}, that they can be constructed by our algorithm as they fall into the set $N_0$. Therefore we can assume for the remainder of the proof, that not every vertex of the translational net contains at least a pair of collinear opposite edges. According to Theorem  \ref{thm:spezi} the corresponding flex can only be 1-dimensional. 

In general there are only a finite number of configurations where two adjacent faces of a $(2\times 2)$ net are coplanar. If this holds true for every  $(2\times 2)$ subnet of the $(m\times n)$ net, then there exist infinitely many configurations during the flex where no two adjacent faces are coplanar. 
Therefore we remain with the discussion of the case where a translational $(2\times 2)$ net has infinitely many configurations where two adjacent faces $A_0$ and $B_0$ are coplanar. 
 As the roles of the generatrix and trajectory are interchangeable, we can assume without loss of generality that $A_0$ and $B_0$ are neighbouring along the generatrix. The other two faces of the $(2\times 2)$ net are denoted by $A_1$ and $B_1$, the hinge between $A_0$ and $A_1$ by $a_0$ and the one between $B_0$ and $B_1$ by $b_0$, respectively. We distinguish the following two cases: 

\begin{enumerate}[(I)]
\item
 Let us assume that $A_0$ and $B_0$ are coplanar during the whole 1-dimensional flex. Therefore we can assume that they are rigidly aligned. Two distinct cases must be considered:
\begin{enumerate}[(i)] 
        \item 
        $a_0$ and $b_0$ are not collinear: The rigidity of the alignment of the faces $A_0$ and $B_0$ propagates along the complete translational $(2\times n)$-strip. As the trajectory cannot be a straight line, this strip would rigidify the complete translational net, which is a contradiction to our assumption. Therefore case (i) is not possible. 
        \item 
        $a_0$ and $b_0$ are collinear: 
        As a consequence the union of $A_0$ and $B_0$ form a parallelogram $C_0$. Thus the complete translational $(2\times n)$-strip can be replaced by a translational $(1\times n)$-strip. By a repeated application of this replacement operations, we can eliminate all cases where adjacent faces are coplanar during the whole flex\footnote{\label{fn:collinear} This replacement operation corresponds to the elimination of collinear vertices in the input polylines.} and we end up with a translational $(k\times l)$ net with $0<k\leq m$, $0<l\leq n$ and $(k,l)\neq (1,1)$. If this net has a configuration during the flex without two adjacent faces coplanar we are done, as we can also construct it with the above algorithm, but where some vertices of the input polylines are collinear (cf.\ Footnote \ref{fn:collinear}). If such a configuration does not exist we proceed with item (II).
    \end{enumerate}
\item 
Due to item (Iii) we can assume that the translational $(k\times l)$ net with $0<k\leq m$, $0<l\leq n$ and $(k,l)\neq (1,1)$ does not have adjacent faces, which are coplanar during the whole flex, but it still does not possess a configuration without two adjacent faces coplanar.  
Theoretically this is possible as a 1-dimensional flex can have different branches. The simplest example for this is a $(2\times 2)$ net, where both pairs of opposite edges are collinear in the flat configuration, which is also the bifurcation configuration. 

It is well known that a bifurcation into different motion modes, where in one of them two faces remain coplanar, can only happen in a flat state of a $(2\times 2)$ net, which has in this state at least on pair of opposite edges collinear. As every of these $(2\times 2)$ nets can only have finitely many of these states, the $(k\times l)$ net, has to have infinitely many configurations during the flex where no $(2\times 2)$ net is in a bifurcation configuration. 
Without loss of generality we can assume that the translational $(k\times l)$ net is in such a configuration. 
Let us assume that the adjacent faces $A_0$ and $B_0$ are coplanar along the complete branch (motion mode) the considered configuration  belongs to. We can now follow the same line of reasoning as done in item (I) with the sole difference that the flex is restricted to a branch. In this way we will end up with a translational $(e\times f)$ net with $0<e\leq k$, $0<f\leq l$ and $(e,f)\neq (1,1)$, which can flex along the branch into a configuration with no adjacent faces coplanar, as there are only finitely many configurations left where two adjacent faces are coplanar.  This configuration can be reconstructed with the above algorithm, but again some vertices of the input polylines are collinear (cf.\ Footnote \ref{fn:collinear}). $\phm$ \hfill $\BewEnde$
\end{enumerate}

\begin{remark} 
The given 3-step algorithm and Theorem \ref{thm:every} show that only $N_1$ can have smooth counterparts and that only $N_0$ and $N_1$ can have semi-discrete counterparts, as solely in these classes both polylines or one of them can be smoothed. 
\hfill $\diamond$
\end{remark}

\subsection{Flexible translational tubes with quadrilateral cross-sections}
\label{sec:tube}

We consider a flexible translational $(m\times 4)$ net, where the generatrix polyline is closed to a quadrilateral polygon. For an arbitrary choice of the generatrix this closing constraint will rigidify the resulting translational tube.  
In the following we describe the set $T_i$ of flexible translational tubes with a quadrilateral cross-section which can be constructed within $N_i$: 
\begin{enumerate}
\item [$T_0$:] 
Within $N_0$ one can construct the trivial set $T_0$ of flexible translational tubes by choosing a parallelogramic profile (see e.g.\ \cite{tachi_parallel}). 
\item [$T_1$:]
Within $N_1$ one can construct the set $T_1$ of flexible translational T-tubes  according to \cite{tubes}, where the problem is discussed in detail. 
\item [$T_3$:]
Within $N_3$ we can only obtain flexible translational tubes which are already contained in $T_0$. The reason for this is as follows: In $F_3$ we have only three directions for generating the closed quadrilateral generatrix. Using all three directions would result in a trapezoidal shape, which is always planar. This implies two contradictions: (i) One can argue that the three spherical points of a set are not located on a great circular arc during the flex of $F_3$, which contradicts the planarity condition. (ii) It can also easily be seen that during the flex of a prismatic tube with trapezoidal cross-section in the given configuration, the parallel faces can only remain parallel for a special trapezoid, namely a parallelogram. 
\item [$T_2$:]
The set $T_2$ of flexible translational tubes constructed within $N_2$ does not only consist of tubes, which are already contained in $T_0$,  but also of novel ones described next: 
For the same reasons as in the last item we cannot construct tubes using three different directions for the quadrilateral generatrix. 
But $F_2$ has four directions $\Vkt q_1,\ldots,\Vkt q_4$, which are symmetric with respect to the orthogonal planes $\alpha$ and $\beta$. Thus we can construct the generatrix as a skew rhombus, which can be done in different ways, where one can distinguish the following two cases: 
    \begin{enumerate}[$\star$]
        \item The skew rhombus has two symmetry planes which are parallel to $\alpha$ and $\beta$.
        \item The skew rhombus has two symmetry planes where on is either parallel to $\alpha$ or $\beta$ and the other one is orthogonal to  $\alpha$ and $\beta$.
    \end{enumerate}
 As the symmetry properties of $F_2$ are preserved during its flex, the skew rhombus remains closed during the flex of the translational tube. 
\end{enumerate}

\begin{remark}
Note that flexible translational tubes with a quadrilateral cross-section can only arise from the composition of flexible bi-prisms with a quadrilateral equator classified by the author \cite{naw_prism}. The flexible bi-prisms composing the novel tubes of the set $T_2$ are either of Type I or Type IIi according to the classification given in \cite[Theorem 6]{naw_prism}. \hfill $\diamond$
\end{remark}

\subsection{Flexible translational polyhedral doubly quadrilateral toroids}\label{sec:toroid}

In a recent joint-work of the author \cite{new_woven} flexible polyhedral doubly quadrilateral (PDQ) toroids are studied, that are, polyhedral toroids with both a quadrilateral cross-section and a quadrilateral trajectory.
Flexible translational PDQ-toroids can be obtained by closing also the trajectory of the tubes discussed in Section \ref{sec:tube} to a quadrilateral polygon. Again for an arbitrary choice of the trajectory this closing constraint will rigidify the resulting PDQ-toroid.  
In the following we describe the set $Q_i$ of flexible translational PDQ-toroids which can be constructed from $T_i$:
\begin{enumerate}
\item [$Q_0$:] A trajectory of a tube from $T_0$ can only remain closed under the flex for the same reason as the generatrix does; thus both polylines have to be parallelograms. The resulting flexible PDQ-toroid is also discussed from other perspectives in \cite{new_woven}. 
\item [$Q_1$:]
Within the set $T_1$ we can construct flexible translational T-toroids, which were studied in detail in \cite{kiumars2}. 
\item [$Q_3$:]
Clearly from $T_3$ we can only obtain flexible PDQ-toroids which are already contained in $Q_0$.
\item [$Q_2$:]
The set $Q_2$ of flexible translational PDQ-toroids constructed within $T_2$ does not only consist of toroids, which are already contained in $Q_0$,  but also of two other types:
    \begin{enumerate} [$\star$]
        \item If the trajectory is a parallelogram then we end up with flexible translational PDQ-toroids, which can also be generated with the approach discussed in \cite{new_woven}. 
        \item 
        The trajectory can also be constructed as a skew rhombus, in the same ways already discussed for  $T_2$. If we combine any of the possible skew-rhombic trajectories with any of the possible skew-rhombic generatrices, then we always end up with a novel flexible translational PDQ-toroid. 
    \end{enumerate}
\end{enumerate}

\section{Conclusion and future research}\label{sec:conclusion}

We gave a full list of translational nets  
which flex within the class of discrete surfaces of translation, by reducing the classification problem to the one of flexible complete bipartite frameworks on the sphere (cf.\ Section \ref{sec:equi}), for which the solution is known. 
The presented list contains the {\it trivially flexible translational nets} $N_0$;  {\it translational T-hedral nets} $N_1$, {\it Bottema's spherical 16-bar based flexible translational nets} $N_2$ and {\it constant diagonal angle flexible translational nets} $N_3$, where the last two classes are novel. 
We also gave a 3-step algorithm for the construction of these translational $(m\times n)$ nets in Section \ref{sec:construct}.  
Moreover, the class $N_2$ implies also novel flexible translational tubes (see $T_2$ of Section \ref{sec:tube}) and novel flexible translational PDQ-toroids (see $Q_2$ of Section \ref{sec:toroid}). Their applicability for the design of flip-flop structures and metamaterials is dedicated to future research. It would also be of interest to identify the special cases of Izmestiev's classification \cite{Izm17} the $(3\times 3)$ subnets of $N_2$ and $N_3$ belong to.

\begin{acknowledgement}
This research was funded in whole or in part by the Austrian Science Fund (FWF) [grant DOI 10.55776/F77]. For open access purposes, the author has applied a CC BY public copyright license to any author accepted manuscript version arising from this submission.
\end{acknowledgement}


\begin{thebibliography}{99.}


\bibitem{Dieleman}	
Dieleman, P., Vasmel, N., Waitukaitis, S., van Hecke, M.: Jigsaw puzzle design of pluripotent origami. Nature Physics \textbf{16}(1):63–68 (2020)

\bibitem{gallet}
Gallet, M., Grasegger, G., Legersky, J., Schicho, J.:
On the existence of paradoxical motions of generically rigid graphs on the sphere. SIAM Journal of Discrete Mathematics \textbf{35}(1):325--361 (2021)

\bibitem{HG18}	
He, Z., Guest, S.D.: On Rigid Origami II: Quadrilateral Creased Papers. 
Proceedings of the Royal Society A \textbf{476}(2237):20200020 (2020)

\bibitem{he26}
He, Z., Hayakawa, K., Ohsaki, M.: Infinitely Refinable Generalization of Quad-Mesh Rigid Origami: From Linear and Equimodular Couplings. 
ASME J. of Mechanisms and Robotics \textbf{18}(3):031004 (2026)

\bibitem{huffman}
Huffman, D.A.: Curvature and creases: a primer on paper. 
IEEE Transactions on Computers \textbf{C-25}:1010--1019 (1976)

\bibitem{Izm17}	
Izmestiev, I.: Classification of flexible Kokotsakis polyhedra with quadrangular base. 
International Mathematics Research Notices \textbf{2017}:715--808 (2017) 

\bibitem{sobhan}
 Izmestiev, I., Nawratil, G., Samareh Rad, S., Sharifmoghaddam, K.: 
 Rigidity and flexibility of discrete conjugate nets with flexible $3\times 3$-subnets. arXiv:2606.15907 (2026)

\bibitem{ivan2}
Izmestiev, I., Raffaelli, M., Rasoulzadeh, A.: 
Smooth, discrete, and semi-discrete Voss surfaces. (in preparation)

\bibitem{ivan1}
Izmestiev, I., Rasoulzadeh, A., Tervooren, J.: 
Isometric Deformations of Discrete and Smooth T-surfaces. 
Computational Geometry \textbf{122}:102104 (2024)

\bibitem{kilian}
Kilian, M., Nawratil, G., Raffaelli, M., Rasoulzadeh, A., Sharifmoghaddam, K.: Interactive design of discrete Voss nets and simulation of their rigid foldings. Computer Aided Geometric Design \textbf{111}:102346 (2024)

\bibitem{kokotsakis}
Kokotsakis, A.: \"Uber bewegliche Polyeder. Mathematische Annalen \textbf{107}:627--647 (1932)

\bibitem{kovalev}
Kovalev, M.D., Orevkov, S.Y.: Complete bipartite graphs flexible in the plane.
Sbornik: Mathematics \textbf{214}(10):1390--1414 (2023)

\bibitem{new_woven}
Mundilova, K., Sharifmoghaddam, K., Nawratil, G.:
Switchable flip-flops from flexible polyhedral doubly quadrilateral toroids. 
(in preparation) 

\bibitem{nassar}
Nassar, H.: Isometric deformations of surfaces of translation. 
Mathematics and Mechanics of Complex Systems \textbf{12}(1):1--17 (2024)

\bibitem{naw_prism}
Nawratil, G.: Flexible octahedra in the projective extension of the {E}uclidean 3-space. Journal for Geometry and Graphics \textbf{14}(2):147--169 (2010)

\bibitem{Naw11}	
Nawratil, G.: Reducible compositions of spherical four-bar linkages with a spherical coupler component. 
Mechanism and Machine Theory \textbf{46}(5):725--742 (2011)

\bibitem{Naw12}	
Nawratil, G.: Reducible compositions of spherical four-bar linkages without a spherical coupler component. 
Mechanism and Machine Theory \textbf{49}:87--103 (2012)

\bibitem{phedra1}
Nawratil, G: From axial C-hedra to general P-nets. Advances in Robot Kinematics, pages 340--347, Springer (2024)

\bibitem{phedra2}
Nawratil, G.: Construction and deformation of P-hedra using control polylines. Advances in Robot Kinematics 2026 - Innovations in Motion: Shaping the Future of Robotic Systems, in press, Springer (2026)

\bibitem{NS10}	
Nawratil, G., Stachel, H.: Composition of spherical four-bar-mechanisms. 
New Trends in Mechanisms Science -- Analysis and Design, pages 99--106, Springer (2010) 

\bibitem{peterson}
Peterson, K:
Ueber {C}urven und {F}l\"achen. 
A. Lang's Buchhandlung/Franz Wagner, Moskau/Leipzig (1868)

\bibitem{sauer}
Sauer, R.: Differenzengeometrie. Springer (1970)

\bibitem{sauer_graf}
Sauer, R., Graf, H.: \"Uber Fl\"achenverbiegung in Analogie zur Verknickung offener Facettenflache. Mathematische Annalen \textbf{105}:499--535 (1931)

\bibitem{BHS08}
Schief, W.K., Bobenko, A.I., Hoffmann, T.: On the integrability of infinitesimal and finite deformations of polyhedral surfaces. 
Discrete Differential Geometry, Oberwolfach Seminars \textbf{38}:67--93 (2008) 

\bibitem{tubes}
Sharifmoghaddam, K., Maleczek, R., Nawratil, G.: Generalizing rigid-foldable tubular structures of T-hedral type. Mechanics Research Communications \textbf{132}:104151 (2023) 

\bibitem{kiumars2}
Sharifmoghaddam, K., Mundilova, K., Nawratil, G., Tachi, T.: Woven Rigidly Foldable T-hedral Tubes Along Translational Surfaces. Origami8, Volume I, pages 141--156, Springer (2026) 

\bibitem{kiumars3}
Sharifmoghaddam, K., Nawratil, G., Rasoulzadeh, A., Tervooren, J.: Using Flexible Trapezoidal Quad-Surfaces for Transformable Design. Proc.\ of IASS Annual Symposia, IASS 2020/21 Surrey Symposium: Transformable structures (IASS WG 15), pages 3236--3248, IASS (2021) 

\bibitem{stachel2}
Stachel, H.: 
A kinematic approach to Kokotsakis meshes.
Computer Aided Geometric Design 
\textbf{27}(6):428--437 (2010)

\bibitem{tile}
Stachel, H.: A Flexible Planar Tessellation with a Flexion Tiling a Cylinder of Revolution. Journal for Geometry and Graphics \textbf{16}(2):153--170 (2012)

\bibitem{tachauer}
Tachauer, A.: \"Uber diejenigen {R}otationsfl\"achen, auf denen zwei {S}charen geod\"atischer {L}inien ein konjugiertes {S}ystem bilden. 
Archiv der Mathematik und Physik \textbf{6}:60--84 (1903)

\bibitem{tachi_parallel}
Tachi, T.:
Composite rigid-foldable curved origami structure. 
Proceedings of Transformables 2013, 6 pages (2013)

\bibitem{voss_orig}
Voss, A.: \"Uber diejenigen {F}l\"achen, auf denen geod\"atische {L}inien ein konjugiertes {S}ystem bilden. 
Münchener Berichte \textbf{18}:95--102 (1888)

\bibitem{voss_ency}
Voss, A.: Abbildung und {A}bwicklung zweier {F}l\"achen aufeinander. Encyklop\"adie der mathematischen Wissenschaften mit Einschluss ihrer Anwendungen \textbf{III}(3):355--440 (1903) 


\bibitem{wunderlich}
Wunderlich, W.: On deformable nine-bar linkages with six triple joints.
Proc.\ of Koninklijke Nederlandse Akademie van Wetenschappen \textbf{A79}(3):257--262 (1976) 

\end{thebibliography}
\end{document}